\documentclass{article}
\usepackage{amssymb,amsthm,amsmath}
\usepackage{pstricks,pst-node,pst-plot}
\usepackage{authblk}
\usepackage{enumerate}
\usepackage{todonotes}

\usepackage{url}
 
\title{Squares of $3$-sun-free split graphs} %\\ and representing graphs by set systems}

\author[1]{Van Bang Le}
\author[2]{Andrea Oversberg}
\author[2]{Oliver Schaudt}

\affil[1]{\small Institut f\"ur Informatik, Universit\"at Rostock, Rostock, Germany\\
\texttt{le@informatik.uni-rostock.de}}
\affil[2]{\small Institut f\"ur Informatik, Universit\"at zu K\"oln, K\"oln, Germany\\ 
\texttt{$\{$oversberg,schaudt$\}$@zpr.uni-koeln.de}}

\begin{document}

\newtheorem{theorem}{Theorem}
\newtheorem{lemma}{Lemma}
\newtheorem{proposition}{Proposition}
\newtheorem{observation}{Observation}
\newtheorem{corollary}{Corollary}
\newtheorem{claim}{Claim}
\newtheorem{definition}{Definition}
\newtheorem{conjecture}{Conjecture}
\newtheorem{question}{Question}
\newtheorem*{question*}{Question}

\newcommand{\dist}{\mbox{dist}}

\newcommand{\sswp}{\textsc{square of split graph with property $\mathcal{P}$}} 
\newcommand{\srwp}{\textsc{split square root with property $\mathcal{P}$}}

\newcommand{\C}{\mathcal C} 

\maketitle

\begin{abstract}
\noindent
The square of a graph $G$, denoted by $G^2$, is obtained from $G$ by putting an edge between two distinct vertices whenever their distance is two. Then $G$ is called a square root of $G^2$. 
Deciding whether a given graph has a square root is known to be NP-complete, even if the root is required to be a split graph, that is, a graph in which the vertex set can be partitioned into a stable set and a clique.
 
We give a wide range of polynomial time solvable cases for the problem of recognizing if a given graph is the square of some special kind of split graph.
To the best of our knowledge, our result properly contains all previously known such cases.
Our polynomial time algorithms are build on a structural investigation of graphs that admit a split square root that is 3-sun-free, and may pave the way toward a dichotomy theorem for recognizing squares of (3-sun-free) split graphs.

\medskip
\noindent{\textbf{Keywords:}} Square of graphs, square of split graphs.\\

\noindent{\textbf{2010 MSC:}} 05C75, 05C85.
\end{abstract}

\section{Introduction}
The \emph{$k$-th power} of a graph $G$, written $G^k$, is obtained from $G$ by adding new edges between any two different vertices at distance at most $k$ in $G$. In case $k=2$, $G^2$ is also called the \emph{square} of $G$, and $G$ is called the \emph{square root} of $G^2$.

While every graph has a square, not every graph admits a square root.
In fact, it is NP-complete to decide whether a given graph has a square root, as was shown by Motwani and Sudan~\cite{MS94}. 
Later, Lau and Corneil~\cite{CL04} proved that this decision problem remains NP-complete when restricted to split graph square roots. 

\begin{theorem}[Lau and Corneil~\cite{CL04}]\label{CL04}
It is NP-complete to decide if a given graph is the square of some split graph.
\end{theorem}

The study of split square roots is particularly interesting due to its close connection to the set representation problem, as was observed by Lau and Corneil~\cite{CL04} and further exploited by Milani\v{c} and Schaudt~\cite{MS13}.
It should be noted that deciding if a given graph is the $k$-th power of some split graph is trivial for $k\ge 3$, as only graphs in which every component is a clique are $k$-th powers, $k\ge 3$, of split graphs.

In view of the hardness result in Theorem~\ref{CL04}, the following problems are of interest:

\medskip\noindent
\srwp\\[.7ex]
\begin{tabular}{l l}
{\em Instance:}& A graph $G$.\\
{\em Question:}& Does there exist a split graph $H\in\mathcal{P}$ such that $G=H^2$\,?\\[.7ex]
\end{tabular}

\medskip\noindent
\sswp\\[.7ex]
\begin{tabular}{l l}
{\em Instance:}& A graph $G\in\mathcal{P}$.\\
{\em Question:}& Does there exist a split graph $H$ such that $G=H^2$\,?\\[.7ex]
\end{tabular}

\medskip
It is worth mentioning that, for ${\cal P'}\subsetneq \cal P$, a polynomial time algorithm solving \sswp\ would imply that \textsc{square of split graphs with property $\cal P'$} is polynomially solvable, too. But, however, this is not longer true for the other problem: In general, a polynomial time algorithm solving \srwp\ does not imply that \textsc{split square root with property $\cal P'$} is also polynomially solvable.

In order to describe previous results for the two above mentioned problems, we need some notion and definitions. 

\subsection{Definitions and notion}

All considered graphs are finite and simple.
Let $G$ be a graph and $v \in V(G)$.
By $N_G(v)$ we denote the set of neighbors of $v$ in $G$.
The \emph{closed neighborhood} of $v$ in $G$, that is $N_G(v) \cup \{v\}$, we denote by $N_G[v]$. 
A \emph{universal vertex} of $G$ is one that is adjacent to every other vertex of $G$. 
A \emph{clique}, respectively, an \emph{independent set}, in $G$ is a set of pairwise adjacent, respectively, non-adjacent vertices, in $G$. 
For a subset $X \subseteq V(G)$, we denote by $G[X]$ the subgraph induced by $X$.
If two graphs $G$ and $H$ are isomorphic, we may simply write $G \cong H$.

For any graph $H$ we say that $G$ is \emph{$H$-free} if $G$ does not contain an induced subgraph that is isomorphic to $H$; $G$ is \emph{$(H_1,\ldots,H_t)$-free} if $G$ is $H_i$-free for all $1\le i\le t$. A graph class is said to be \emph{hereditary} if whenever a graph belongs to the class then all of its induced subgraphs also belong to the class.

For an integer $\ell\ge 3$, $C_\ell$ denotes the cycle on $\ell$ vertices and $\ell$ edges. 
A graph is \emph{chordal} if it is $C_\ell$-free for all $\ell\ge 4$. 
A chordal graph is {\em strongly chordal} if it does not contain any $\ell$-sun as an induced  subgraph; here, an \emph{$\ell$-sun}, denoted by $S_\ell$, consists of a stable set $\{u_1, u_2, \ldots, u_\ell\}$ and a clique $\{v_1, v_2, \ldots, v_\ell\}$ such that for $i \in \{1,\ldots, \ell\}$, $u_i$ is adjacent to exactly $v_i$ and $v_{i+1}$ (index arithmetic modulo $\ell$). An \emph{odd sun} is an $S_\ell$ with odd $\ell$.

Given two graphs $G$ and $H$, the \emph{join} $G\oplus H$ is obtained from the disjoint union of $G$ and $H$ by adding all possible edges between vertices in $G$ and vertices in $H$. 

A \emph{split graph} is a graph whose vertex set can be partitioned into a clique and an independent set. It is well known that split graphs are exactly the chordal graphs without induced $2K_2$ (the complement of the $4$-cycle $C_4$); $K_n$ stands for a complete graph with $n$ vertices. 

For a graph $G$, $\mathcal{C}(G)$ denotes the set of all inclusion-maximal cliques of $G$. $G$ is said to be \emph{clique-Helly} if $\mathcal{C}(G)$ has the Helly property. 
$G$ is \emph{hereditary clique-Helly} if every induced subgraph of $G$ is  
clique-Helly. (See~\cite{DPS09} for more information on clique-Helly graphs.) 
Prisner~\cite{Prisner93} characterized hereditary clique-Helly graphs as follows; see Fig.~\ref{fig:G1234} for the graphs $G_1, \ldots, G_4$.

\begin{theorem}[Prisner \cite{Prisner93}]\label{clique-Helly}
A graph $G$ is hereditary clique-Helly if and only if $G$ is $(G_1, G_2, G_3, G_4)$-free. 
\end{theorem}

\begin{figure}[ht] %b]
\begin{center}
\psset{unit=0.8}
\begin{pspicture}(0,0)(11,1.9)

\cnode(0,1.7){0.1cm}{a_1}
\cnode(0.5,0.85){0.1cm}{a_2}
\cnode(1,0){0.1cm}{a_3}
\cnode(1,1.7){0.1cm}{a_4}
\cnode(1.5,0.85){0.1cm}{a_5}
\cnode(2,1.7){0.1cm}{a_6}

\ncarc[arcangle=0]{-}{a_1}{a_2}
\ncarc[arcangle=0]{-}{a_1}{a_4}
\ncarc[arcangle=0]{-}{a_2}{a_4}
\ncarc[arcangle=0]{-}{a_2}{a_3}
\ncarc[arcangle=0]{-}{a_2}{a_5}
\ncarc[arcangle=0]{-}{a_3}{a_5}
\ncarc[arcangle=0]{-}{a_4}{a_5}
\ncarc[arcangle=0]{-}{a_4}{a_6}
\ncarc[arcangle=0]{-}{a_5}{a_6}

\cnode(3,1.7){0.1cm}{b_1}
\cnode(3.5,0.85){0.1cm}{b_2}
\cnode(4,0){0.1cm}{b_3}
\cnode(4,1.7){0.1cm}{b_4}
\cnode(4.5,0.85){0.1cm}{b_5}
\cnode(5,1.7){0.1cm}{b_6}

\ncarc[arcangle=0]{-}{b_1}{b_2}
\ncarc[arcangle=-40]{-}{b_1}{b_3}
\ncarc[arcangle=0]{-}{b_1}{b_4}
\ncarc[arcangle=0]{-}{b_2}{b_4}
\ncarc[arcangle=0]{-}{b_2}{b_3}
\ncarc[arcangle=0]{-}{b_2}{b_5}
\ncarc[arcangle=0]{-}{b_3}{b_5}
\ncarc[arcangle=0]{-}{b_4}{b_5}
\ncarc[arcangle=0]{-}{b_4}{b_6}
\ncarc[arcangle=0]{-}{b_5}{b_6}

\cnode(6,1.7){0.1cm}{c_1}
\cnode(6.5,0.85){0.1cm}{c_2}
\cnode(7,0){0.1cm}{c_3}
\cnode(7,1.7){0.1cm}{c_4}
\cnode(7.5,0.85){0.1cm}{c_5}
\cnode(8,1.7){0.1cm}{c_6}

\ncarc[arcangle=0]{-}{c_1}{c_2}
\ncarc[arcangle=-40]{-}{c_1}{c_3}
\ncarc[arcangle=0]{-}{c_1}{c_4}
\ncarc[arcangle=0]{-}{c_2}{c_4}
\ncarc[arcangle=0]{-}{c_2}{c_3}
\ncarc[arcangle=0]{-}{c_2}{c_5}
\ncarc[arcangle=0]{-}{c_3}{c_5}
\ncarc[arcangle=-40]{-}{c_3}{c_6}
\ncarc[arcangle=0]{-}{c_4}{c_5}
\ncarc[arcangle=0]{-}{c_4}{c_6}
\ncarc[arcangle=0]{-}{c_5}{c_6}

\cnode(9,1.7){0.1cm}{d_1}
\cnode(9.5,0.85){0.1cm}{d_2}
\cnode(10,0){0.1cm}{d_3}
\cnode(10,1.7){0.1cm}{d_4}
\cnode(10.5,0.85){0.1cm}{d_5}
\cnode(11,1.7){0.1cm}{d_6}

\ncarc[arcangle=0]{-}{d_1}{d_2}
\ncarc[arcangle=-40]{-}{d_1}{d_3}
\ncarc[arcangle=0]{-}{d_1}{d_4}
\ncarc[arcangle=40]{-}{d_1}{d_6}
\ncarc[arcangle=0]{-}{d_2}{d_4}
\ncarc[arcangle=0]{-}{d_2}{d_3}
\ncarc[arcangle=0]{-}{d_2}{d_5}
\ncarc[arcangle=0]{-}{d_3}{d_5}
\ncarc[arcangle=-40]{-}{d_3}{d_6}
\ncarc[arcangle=0]{-}{d_4}{d_5}
\ncarc[arcangle=0]{-}{d_4}{d_6}
\ncarc[arcangle=0]{-}{d_5}{d_6}

\end{pspicture}
\end{center}
\caption{$G_1$, $G_2$, $G_3$, and $G_4$.}
\label{fig:G1234}
\end{figure}
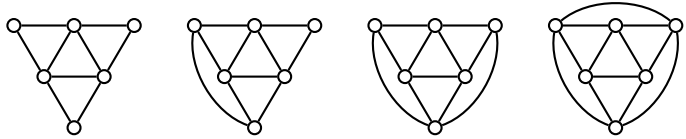

It follows that a split graph is hereditary clique-Helly if and only if it is $3$-sun-free. 
For more information on graph classes, their definitions and properties we refer to the books~\cite{BLS99,Spinrad} and the online resource~\cite{isgci}.

Finally, as the square of a graph is the union of the squares of its components, we may assume that all graphs considered in this paper are connected.

\subsection{Previous results}\label{sec:state-of-art}

%%%%%%%%%%%%%%%%%%%%%%%%%%
%%%%%%%%%%%%%%%%%%%%%%%%%%

For split graphs $H=(V(H), E(H))$ we write $H=(C\cup I, E(H))$, meaning $V(H)=C\cup I$ is a partition of the vertex set of $H$ into a clique $C$ and an independent set $I$.
The following facts, Propositions~\ref{prop:1} and \ref{prop:2}, are proved implicitly in \cite{LT10a,LT11,LOS14}. We write $\bigcap \mathcal{C}(G)$ shortly for $\bigcap_{Q \in \mathcal C(G)} Q$.

\begin{proposition}\label{prop:1}
If $\big|\bigcap \mathcal{C}(G) \big|\ge |\mathcal{C}(G)|$, then $G=H^2$ for some split graph $H$. Moreover, this split square root $H=(C\cup I, E(H))$ can be constructed as follows: % such that 
% \begin{itemize}
%  \item[\em (i)] If $G$ is strongly chordal, then so is $H$.
%  \item[\em (ii)] If $G$ is hereditary clique-Helly, then so is $H$.
% \end{itemize}
\begin{itemize}
 \item Write $\mathcal{C}(G)=\{Q_1,\ldots, Q_q\}$ and set $C=\bigcap \mathcal{C}(G)$, say $C=\{c_1,c_2,\ldots, c_q$, $\ldots\}$, and 
$I=V(G)-C$.
 \item There is an edge $xc_i$ in $H$ between $x\in I$ and $c_i\in C$ if and only if $x\in Q_i$ in $G$.
\end{itemize}
\end{proposition}

We note that the sufficient condition in Proposition~\ref{prop:1} is not necessary. Indeed, the join of the line graph $L(K_n)$ and $K_n$, $G=L(K_n)\oplus K_n$, is the square of a split graph, but $\big|\bigcap \mathcal{C}(G)\big|=n < |\mathcal{C}(G)|={n \choose 3} + n$.

\begin{proposition}\label{prop:2}
Let $H=(C\cup I, E(H))$ be a split graph without induced $3$-sun. Then $Q$ is 
a maximal clique in $H^2$ if and only if $Q=N_H[v]$ for some vertex 
$v\in C$ with inclusion-maximal $N_H[v]$. In particular, if $G=H^2$ for some $3$-sun-free split graph $H$, then $G$ has at most $|V(G)|$ maximal cliques and $\big|\bigcap \mathcal{C}(G) \big|\ge |\mathcal{C}(G)|$.
\end{proposition}

Based on Propositions~\ref{prop:1} and \ref{prop:2}, the following polynomially solvable cases of \srwp\ have been obtained.

\begin{theorem}[Le and Tuy~\cite{LT10a,LT11}]\label{LT11Char}
A graph $G$ is the square of a strongly chordal split graph if and only if $G$ is strongly chordal and $\big|\bigcap \mathcal{C}(G) \big|\ge |\mathcal{C}(G)|$.
\end{theorem}

This theorem has the following algorithmic implication.
\begin{theorem}[Le and Tuy~\cite{LT10a,LT11}]\label{LT11Algo}
Given an $n$-vertex and $m$-edge graph $G$, recognizing if $G$ is the square of some 
strongly chordal split graph $H$ can be done in time $\mathcal O(\min\{n^2,m\log n\})$, and if so, such a square root $H$ for $G$ can be constructed in the same time.
\end{theorem}

Subsequently, we considered the case of 3-sun-free split graphs.

\begin{theorem}[Le, Oversberg and Schaudt~\cite{LOS14}]\label{thm:squareofs3freesplit}
A graph $G$ is the square of a connected 3-sun-free split graph if and only if $G$ is $(G_1, G_2, G_3, G_4)$-free 
and satisfies $\big|\bigcap \mathcal{C}(G) \big|\ge |\mathcal{C}(G)|$.
\end{theorem}

With Theorem~\ref{clique-Helly}, we obtained the following reformulation of Theorem~\ref{thm:squareofs3freesplit}.

\begin{theorem}[Le, Oversberg and Schaudt~\cite{LOS14}]\label{thm:squareofs3freesplit-v2}
A graph $G$ is the square of a connected hereditary clique-Helly split graph if and only if $G$ is a hereditary clique-Helly graph satisfying $\big|\bigcap \mathcal{C}(G) \big|\ge |\mathcal{C}(G)|$.
\end{theorem} 

As a consequence of Theorem~\ref{thm:squareofs3freesplit}, we derived the following result.

\begin{theorem}[Le, Oversberg and Schaudt~\cite{LOS14}]\label{thm:recogsquareofs3freesplit}
It can be decided in $\mathcal O(n^2m)$ time whether a given $n$-vertex $m$-edge graph has a $3$-sun-free split square root, 
and if so, such a square root can be constructed in the same time.
\end{theorem}

In \cite{MS13} the following polynomially solvable case of \sswp\ is proved.

\begin{theorem}[Milani\v c and Schaudt~\cite{MS13}]\label{MS13}
It can be decided in linear time whether a given chordal graph has a split square root, and if so, such a square root can be constructed in the same time.
\end{theorem}

\subsection{Our Contributions}
%%%%%%%%%%%%%%%%%%%%%%%%%%%%%%

Our first and main result provides a general framework for obtaining polynomially solvable cases of \srwp, in the case when the 3-sun is forbidden for the property $\mathcal P$.
For this, we introduce the following notion.

\begin{definition}
Let $\mathcal G$ be a hereditary class of $3$-sun-free split graphs.
We call $\mathcal G$ \emph{simple} if there is a function $f: \mathcal G \to \mathcal G$ such that $f(H)^2 \cong H^2 \oplus K_1$, for all $H \in \mathcal G$.
\end{definition}

We remark that the function $f$ in the definition of simple graph classes is not needed to be computable in polynomial time.

Loosely speaking, $\mathcal G$ is simple if the graphs in $\mathcal G$ can be augmented by one vertex which appears as a universal vertex in the square while the rest of the square remains unchanged.
As we see below, this condition is met by several natural classes of split graphs.

\begin{theorem}\label{thm:recognition}
Let $\mathcal G$ be a simple hereditary class of $3$-sun-free split graphs.
Assuming $\mathcal G$ is polynomially recognizable, squares of graphs of $\mathcal G$ can be recognized in polynomial time.
\end{theorem}

Next we give a list of examples of classes to which Theorem~\ref{thm:recognition} can be successfully applied.
To the best of our knowledge, all previously known polynomial time solvable cases for \srwp\ are contained in this list. (The complement of a $3$-sun is also called \emph{net}. See the Appendix for more information about the graph classes listed below.)

\begin{corollary}\label{cor:classes}
For the following graph classes $\mathcal G$ it is decidable in polynomial time whether a given connected graph has a square root in $\mathcal G$.
If existent, a square root of $G$ in $\mathcal G$ can be computed in polynomial time, too.
\begin{enumerate}[\em (a)]
	\item $3$-sun-free split graphs;
	\item $(\mbox{$3$-sun, net})$-free split graphs;
	\item strongly chordal split graphs;
    \item odd-sun-free split graphs; %(in $\mathcal O(m^{18})$ time).
	\item interval split graphs; 
	\item permutation split graphs; %(in $\mathcal O(\max\{n^2m,m^2\})$ time);
	\item comparability split graphs; %(in $\mathcal O(\max\{n^2m,MM\})$ time);
	\item probe threshold split graphs.
\end{enumerate}
\end{corollary}

We complement this algorithmic result by proving a structural characterization of the squares of graphs contained in some of these classes.
For simple classes $\mathcal G$ of 3-sun-free split graphs, defined by finitely many forbidden induced subgraphs, we have the following result.
It shows that the graphs admitting a square root in $\mathcal G$ are essentially characterized by a finite list of forbidden induced subgraphs.

\begin{theorem}\label{thm:simple-and-finite}
Let $\mathcal G$ be a simple hereditary class of $3$-sun-free split graphs, defined by a finite set $\mathcal F$ of forbidden induced subgraphs.
There is a finite set $\mathcal F'$ of graphs such that the following holds:
a given connected graph $G$ has a square root in $\mathcal G$ if and only if $|\bigcap \mathcal C(G)| \ge |\mathcal C(G)|$ and $G$ is $\mathcal F'$-free.
\end{theorem}

In particular, Theorem~\ref{thm:simple-and-finite} applies to the classes of (3-sun, net)-free split graphs, the largest self-complementary hereditary subclass of 3-sun-free split graphs.
It also applies to the class of interval split graphs, permutation split graphs, comparability split graphs, and probe threshold split graphs.
Let us stress the fact that Theorem~\ref{thm:simple-and-finite} is solely an existence statement, and it seems to be difficult to derive explicit such characterizations, say for any of the classes mentioned above.

%While Theorem~\ref{thm:simple-and-finite} is only an existence result, we now explicitely derive such a characterization which was previously unknown.
%In our next theorem we consider the case of $(S_3,\mbox{net})$-free split graphs, the largest self-complementary hereditary subclass of 3-sun-free split graphs.

%\begin{theorem}\label{thm:nets3free}
%A graph $G$ is the square of an $(S_3,\mbox{net})$-free split graph if and only if $G$ is $(G_1, G_2, G_3, G_4, \mathrm{net}^2)$-free 
%and satisfies $\big|\bigcap_{Q \in \mathcal{C}(G)} Q\big|\ge |\mathcal{C}(G)|$.
%\end{theorem}

Our last result concerns an important class of split graphs which is covered by Corollary~\ref{cor:classes} but not by Theorem~\ref{thm:simple-and-finite}, namely odd-sun-free split graphs.
For this, we need the following notions.
%A 0/1-matrix is \emph{balanced} if it does not contain an odd square submatrix with exactly two ones per row and column as a submatrix.
A hypergraph is \emph{balanced} if and only if its vertex-to-hyperedge incidence matrix is balanced.
It is known that a hypergraph is balanced if and only if its vertex-to-hyperedge incidence graph does not contain induced cycles of length $2k$, for any odd $k$.
Bonomo et al.~\cite{BDLS06} introduced the class of balanced graphs as the class of graphs whose maximal clique hypergraph is balanced.

\begin{theorem}\label{thm:odd-sun-free-characterization}
A graph $G$ is the square of a connected odd-sun-free split graph if and only if it is balanced and $\big|\bigcap \mathcal{C}(G)\big|\ge |\mathcal{C}(G)|$.
In particular, there is a polynomial time algorithm to recognize graphs that admit an odd-sun-free split square root.
\end{theorem}

We remark that, in contrast to strongly chordal split graphs and 3-sun-free split graphs, the recognition of odd-sun-free split graphs is highly involved, 
as it requires the decomposition of balanced matrices~\cite{CCR99}.

In the subsequent sections~\ref{sec:computing} and~\ref{sec:characterizations} we prove the above results.
We remark that the concepts we introduce in the proof of Theorem~\ref{thm:recognition} lead to fairly short proofs of the other results.

We close the paper in Section~\ref{sec:conclusion} with a short discussion of our results and related future work.
There we also provide questions for further research which may lead to a deeper understanding of the problem of finding split square roots.

%%%%%%%%%%%%%%%%%%%%%%%%%%%%%%%%%%%%%%%%%%%%%%%%%%%%%%%%%%%%%%%%%%%%%%%%%%%%%%%%%%
\section{Proofs of Theorem~\ref{thm:recognition} and Corollary~\ref{cor:classes}}\label{sec:computing}
%%%%%%%%%%%%%%%%%%%%%%%%%%%%%%%%%%%%%%%%%%%%%%%%%%%%%%%%%%%%%%%%%%%%%%%%%%%%%%%%%%

To prepare the proof of Theorem~\ref{thm:recognition} and Corollary~\ref{cor:classes}, we need the following concept, which modifies the construction of the square root in Proposition~\ref{prop:1}.

Let $G$ be a graph with $\big|\bigcap \mathcal{C}(G) \big| \ge |\mathcal{C}(G)|$.
%Following Propositions~\ref{prop:1} and~\ref{prop:2}, we define the \emph{trunk} of $G$ to be the split graph $T(G)=(C\cup I, E)$ such that
%\begin{itemize}
 %\item $C= \mathcal{C}(G)$, say $C=\{Q_1,Q_2,\ldots, Q_q,\ldots\}$, where $q=|\mathcal{C}(G)|$, and $I=V(G)-\bigcap \mathcal C$, and
 %\item there is an edge $xQ_i$ in $T(G)$ if and only if $x\in Q_i$ in $G$.
%\end{itemize}
Let $B$ be the bipartite graph with the bipartition $V(B)=\left(V(G)\setminus\bigcap\mathcal{C}(G)\right)\cup \mathcal C(G)$ and edge set $E(B)= \{vQ \mid v \in V(G)\setminus\bigcap\mathcal{C}(G), Q \in \mathcal{C}(G), v \in Q\}$. 
Note that in case $G$ is non-complete, $|\mathcal{C}(G)|=|\mathcal{C}\left(G\setminus\bigcap\mathcal{C}(G)\right)|$, hence $B$ is (isomorphic to) the bipartite vertex-to-maximal-clique incidence graph of $G\setminus\bigcap\mathcal{C}(G)$.

The \emph{trunk} of $G$, denoted $T(G)$, is the graph obtained from $B$ by including all possible edges between the vertices of $B$ representing the maximal cliques.
That is, $T(G)$ is the split graph with $V(T(G)) = C \cup I$, where $C=\mathcal{C}(G)$ is a clique and $I = V(G) \setminus \bigcap\mathcal{C}(G)$ is an independent set.
Since $\big|\bigcap \mathcal{C}(G) \big| \ge |\mathcal{C}(G)|$, we may at times identify $C$ with some subset of $\bigcap \mathcal C(G)$.
In particular, we may treat $T(G)$ as a subgraph of a certain square root of $G$. (Indeed, $T(G)$ is an induced subgraph of the split square root $H$ in Proposition~\ref{prop:1}.)

Proposition~\ref{prop:2} implies the following.

\begin{lemma}\label{lem:trunk}
\mbox{}
\begin{enumerate}[\em (a)]
	\item\label{st:a} For any graph $G$ with $\big|\bigcap \mathcal{C}(G) \big| \ge |\mathcal{C}(G)|$, $G \cong T(G)^2 \oplus K_r$, where $r:=\big| \bigcap \mathcal{C}(G) \big| - |\mathcal{C}(G)|$.
	\item\label{st:b} Let $G=H^2$ for some connected $3$-sun-free split graph $H=(C\cup I,E)$. Then $\big| \bigcap \mathcal{C}(G) \big| \ge |\mathcal{C}(G)|$ and $H$ contains $T(G)$ as an induced subgraph. Moreover, if $\big| \bigcap \mathcal{C}(G) \big| = |\mathcal{C}(G)|$, then $H\cong T(G)$.
\end{enumerate}
\end{lemma}
\begin{proof}
Statement (\ref{st:a}) follows immediately from the definition of $T(G)$ and the following facts. (1) Two vertices in $V(G)\setminus \bigcap \mathcal{C}(G)$ are adjacent in $G$ if and only if they belong to a common maximal clique, i.e., they have a common neighbor in $T(G)$, and (2) $\bigcap \mathcal{C}(G)$ is exactly the set of all universal vertices of $G$ and $\mathcal{C}(G)$ is exactly the set of all universal vertices of $T(G)^2$. 

The first part of statement (\ref{st:b}) follows directly from Proposition~\ref{prop:2}. By Proposition~\ref{prop:2} again, 
for every $Q \in \mathcal C(G)$ there is a unique vertex $v_Q\in C$ with $Q=N_H[v_Q]$ and such that $v_Q\not= v_{Q'}$ for $Q\not=Q'$.
Thus, letting $C'=\{v_Q\in C\mid Q\in \mathcal{C}(G)\}$, we have
\[
C'\subseteq C\, \mbox{ and }\, |C'|=|\mathcal{C}(G)|.
\]
Next, letting $I' := V(G) \setminus \bigcap \mathcal{C}(G)$, we have
\[
I' \subseteq I,
\]
because no vertex of $I'$ is universal in $G$, while all vertices in $C$ are universal in $H^2=G$.

Now, by definition of $C'$, a vertex $u\in I'$ is adjacent, in $H$, to a vertex $v_Q\in C'$ if and only if $u\in Q$. This implies that 
\[
T(G)\cong H[C' \cup I'].
\] 
Moreover, if $\big| \bigcap \mathcal{C}(G) \big| = |\mathcal{C}(G)|$, then clearly, $H[C' \cup I'] = H$.
%Hence statement (\ref{st:b}).

\end{proof}

Given a simple hereditary class of $3$-sun-free split graphs $\cal G$, Lemma~\ref{lem:trunk} enables us to characterize graphs having a square root in $\cal G$ as follows.

\begin{lemma}\label{lem:simple}
Let $\mathcal G$ be a simple hereditary class of $3$-sun-free split graphs, and let $G$ be a connected graph. 
Then $G$ has a square root in $\mathcal G$ if and only if 
\begin{equation}\label{eqn:root-or-not}
\big| \bigcap \mathcal{C}(G) \big| \ge |\mathcal{C}(G)| \mbox{ and } T(G) \in \mathcal G.
\end{equation}
\end{lemma}
\begin{proof}
As Lemma~\ref{lem:trunk} (\ref{st:b}) states, if $G$ admits a 3-sun-free square root $H$, then $\big| \bigcap \mathcal{C}(G) \big|$ $\ge |\mathcal{C}(G)|$ and $H$ contains $T(G)$ as an induced subgraph. If $H\in \cal G$, then $T(G)\in \cal G$ because $\cal G$ is hereditary. 
Thus, if (\ref{eqn:root-or-not}) fails to hold, $G$ does not have a square root in $\mathcal G$.

To see the converse, assume that (\ref{eqn:root-or-not}) holds, and let $r:=\big| \bigcap \mathcal{C}(G) \big| - |\mathcal{C}(G)|$. Because $\cal G$ is simple, we have a function $f: \mathcal G \to \mathcal G$ such that $f(H)^2 \cong H^2 \oplus K_1$, for all $H \in \mathcal G$. Let $f^r=f\circ f \circ \cdots \circ f$ be the $r$-fold composition of $f$; $f^0$ stands for the identity. 
Since $T(G)\in \mathcal G$ and $\mathcal G$ is simple, $f^r(T(G)) \in \mathcal G$. 
Moreover, $f^r(T(G))^2 \cong T(G)^2 \oplus K_r$. Hence, by Lemma~\ref{lem:trunk} (\ref{st:a}), 
$G \cong f^r(T(G))^2$.  
Thus, $G$ has a square root in $\mathcal G$, namely $f^r(T(G))$.
\end{proof}

We now are able to prove our main result.

\begin{proof}[Proof of Theorem~\ref{thm:recognition}.]
Given an arbitrary graph $G$, we may check, in polynomial time, whether the condition~(\ref{eqn:root-or-not}) holds as follows.

First we check if $G$ has a $3$-sun-free split square root at all. This step can be done in time $\mathcal{O}(|V(G)|^2|E(G)|)$ by Theorem~\ref{thm:recogsquareofs3freesplit}. 
If this is not the case, we can correctly output that, in particular, $G$ has no square root in $\mathcal G$.

Otherwise, we know by Proposition~\ref{prop:2} that $\big| \bigcap \mathcal{C}(G) \big| \ge |\mathcal{C}(G)|$. So, it remains to check whether $T(G) \in \mathcal G$.

By Proposition~\ref{prop:2} again, $G$ has at most $|V(G)|$ maximal cliques. 
Thus, using the algorithm in~\cite{TIAS77}, all maximal cliques in $G$ can be listed in time $\mathcal{O}(|V(G)|^2|E(G)|)$, and therefore, $T(G)$ can be constructed in the same time complexity. 
Now, if $\mathcal G$ can be recognized in polynomial time, it can be decided whether $T(G) \in \mathcal G$ in polynomial time. By Lemma~\ref{lem:simple}, this completes the proof.
\end{proof}

%\begin{observation}\label{lem:trunk-size}
%If $G$ is a connected graph with $\big| \bigcap \mathcal{C}(G) \big| \ge |\mathcal{C}(G)|$, $T(G)$ has at most $n+m$ vertices and $2m^2$ many edges.
%\end{observation}

We can now state the proof of Corollary~\ref{cor:classes}.

\begin{proof}[Proof of Corollary~\ref{cor:classes}.]
Let $\mathcal G$ be one of the listed classes, and let $H=(C \cup I,E(H)) \in \mathcal G$.
We define $f(H)$ to be the graph obtained from $H$ by adding a new vertex and making it fully adjacent to $C$. It follows from the definition of $\cal G$ (see also Appendix), that $f(H)\in\cal G$. Clearly, in $f(H)^2$, the new vertex is universal, hence $f(H)^2\cong H^2\oplus K_1$. Thus, $\mathcal G$ is simple.

Note that the recognition problem for all mentioned graph classes is solvable in polynomial time, see Appendix. %cf.~\cite{BLS99}. 
Consequently, we can check whether $T(G) \in \mathcal G$ in polynomial time. (Construction of $T(G)$ as indicated in proof of Theorem~\ref{thm:recognition}.) 
Since $f$ is clearly computable in polynomial time in $|V(H)|+|E(H)|$, a square root of $G$ in $\mathcal G$, namely $f^r(T(G))$ (cf. proof of Lemma~\ref{lem:simple}), can indeed be computed in polynomial time.
This completes the proof.
\end{proof}

\section{Proofs of Theorems~\ref{thm:simple-and-finite} and~\ref{thm:odd-sun-free-characterization}}\label{sec:characterizations}
%%%%%%%%%%%%%%%%%%%%%%%%%%%%%%%%%%%%%%%%%%%%%%%%%%%%%%%%%%%%%%%%%%%

\begin{proof}[Proof of Theorem~\ref{thm:simple-and-finite}]
Let $G$ be a graph with the following properties.
\begin{enumerate}[(a)]
	\item\label{item:forbidden1} $|\bigcap \mathcal C(G)| \ge |\mathcal C(G)|$;
	\item\label{item:forbidden2} $G$ does not admit a square root in $\mathcal G$;
	\item\label{item:forbidden3} every proper induced subgraph $G'$ of $G$ with $|\bigcap \mathcal C(G')| \ge |\mathcal C(G')|$ admits a square root in $\mathcal G$.
\end{enumerate}
We first prove that the order of $G$ is bounded by a constant $N=N(\mathcal F)$ depending only on $\mathcal F$,
namely by 
\[
N := \max\left\{2|V(F)|^2 + 2^{2|V(F)|^2} : F \in \mathcal F\right\}.
\]
After that, we show how the existence of a finite set $\mathcal F'$ as in the statement of the theorem follows.

In order to control the structure of the split square roots of $G$, we first show that we may assume $G$ to have a $3$-sun-free split square root.
Suppose the contrary.
By Theorem~\ref{thm:squareofs3freesplit}, $G$ contains one of $G_1,\ldots,G_4$ from Fig.~\ref{fig:G1234} as an induced subgraph, say $G_i$.
%Thus, $G \cong  G_i \oplus K_r$ for some $r \le 7$.
Hence, $|\bigcap \mathcal C(G)| \ge |\mathcal C(G)| \ge |\mathcal C(G_i)|$, and so $G$ contains $G' := G_i \oplus K_r$ as an induced subgraph, where $r:=|\mathcal C(G_i)|\le 8$.
But all conditions~(\ref{item:forbidden1})--(\ref{item:forbidden3}) apply to $G'$, too, and thus $G \cong G'$.
This means $|V(G)| = |V(G_i)| + |\mathcal C(G_i)| \le 14$, as desired.

So we know that $G$ has a $3$-sun-free split square root.
By Lemma~\ref{lem:trunk}, any such square root contains $T(G) = (C \cup I,E)$ as an induced subgraph. 
Recall that $C=\cal C(G)$ and $I=V(G)\setminus\bigcap \mathcal{C}(G)$.

By Lemma~\ref{lem:simple}, $T(G) \notin \mathcal G$.
Thus, there is an induced subgraph $F$ of $T(G)$ that is isomorphic to some member of $\mathcal F$.
Since $T(G)$ is 3-sun-free and split, $F$ is a 3-sun-free split graph, too.
Our aim is now to identify a small subgraph $G'$ of $G$ whose any $S_3$-free split square root contains $F$ as an induced subgraph.

For this, let $C(F) = V(F) \cap C$ and $I(F) = V(F) \cap I$.
For each $v \in C(F)$ let $S_v = N_{T(G)}(v) \cap I$.
Moreover, for each $x \in I(F) \setminus S_v$ we may choose a vertex $v_x \in S_v$ such that $xv_x \notin E(G)$.
This is possible since $S_v$ forms, together with the universal vertices of $G$, a maximal clique in $G$. 
Let $U_v$ be the collection of all these vertices $v_x$, and let $W := I(F) \cup \bigcup_{v \in C(F)} U_v$.
Note that $|U_v|\le |I(F) \setminus S_v|$, hence
\begin{equation}\label{eqn:size-of-W}
|W| \le |I(F)|+|C(F)| \cdot |I(F)| \le |V(F)|^2.
\end{equation}
Finally, we choose a set $W' \subseteq I$ such that, for each vertex $x \in W$, there is a vertex $x' \in W'$ with $xx' \notin E(G)$.
This is possible since $W \subseteq I$ implies that no vertex of $W$ is universal in $G$.
Note that we do not require $W'$ to be disjoint to $W$.
Clearly we can choose $W'$ such that $|W'| \le |W|$. %, a fact that we make use of below.

Let $C' \subseteq C$ such that $|C'| = |\mathcal C(G[W \cup W'])|$.
Such a set exists since $|C| = |\mathcal C(G)| \ge |\mathcal C(G[W \cup W'])|$.
We now consider the graph $G' := G[W \cup W' \cup C']$.
Since $|W'| \le |W|$, (\ref{eqn:size-of-W}) implies
\begin{align}
\nonumber |V(G')| &\le |W|+|W'|+|\mathcal C(G[W \cup W'])|\\
\nonumber  				&\le 2|V(F)|^2 + 2^{|W|+|W'|}\\
									&\le 2|V(F)|^2 + 2^{2|V(F)|^2} \le N. \label{eqn:size-of-G'}
\end{align}

Note that the set of universal vertices of $G'$ equals $C'$.
Since $\bigcap \mathcal C(G') = C'$, we have $|\mathcal C(G')| = |\bigcap \mathcal C(G')|$.
By Theorem~\ref{thm:squareofs3freesplit}, $G'$ has a 3-sun-free split square root $H =: (C(H) \cup I(H) , E(H))$.
As $|\bigcap \mathcal C(G')| = |\mathcal C(G')|$, $H \cong T(G')$ by Lemma~\ref{lem:trunk}.

We now show that $H$ contains $F$ as an induced subgraph.
Note that, since no vertex of $W \cup W'$ is universal in $G'$, $I(H)=W \cup W'$, and so $C(H)=C'$.
%Also, $C(H)=C'$.
For each $v \in C(F)$, the set $(S_v \cap I(F)) \cup U_v$ is a clique in $G'$.
Thus, there is some $Q_v \in \mathcal C(G')$ with $(S_v \cap I(F)) \cup U_v \subseteq Q_v$.
By the definition of $U_v$, for each $x \in I(F) \setminus S_v$ there is some $v_x \in U_v$ with $xv_x \notin E(G)$ and thus $xv_x \notin E(G')$.
Therefore $(I(F) \setminus S_v) \cap Q_v = \emptyset$.
This in turn means $Q_v \cap I(F) = S_v \cap I(F)$.
Moreover, for each $v \in C(F)$ there is some $v' \in C'$ such that $N_H[v']=Q_v$.
Hence, $H[I(F) \cup \{v' : v \in C(F)\}] \cong F$, as claimed.

This means $H \notin \mathcal G$, and thus $T(G') \notin \mathcal G$. Therefore, by Lemma~\ref{lem:simple}, $G'$ does not admit a square root in $\cal G$. 
By the choice of $G$, $G = G'$.
Since by~(\ref{eqn:size-of-G'}) the order of $G$ is bounded by $N$, the first part of the proof is complete.
\medskip

So we know that the order of every graph that satisfies the conditions~(\ref{item:forbidden1})--(\ref{item:forbidden3}) is bounded by some constant $N$ depending only on $\mathcal F$.
We pick $\mathcal F'$ as the finite set of all such graphs.
In order to prove that $\mathcal F'$ satisfies the statement of the theorem, let $G$ be any connected graph.

First we assume that $G$ satisfies $|\bigcap \mathcal C(G)| \ge |\mathcal C(G)|$ and is, furthermore, $\mathcal F'$-free.
In particular, (\ref{item:forbidden1}) holds for $G$.
Suppose that $G$ does not have a square root in $\mathcal G$, that is, (\ref{item:forbidden2}) holds for $G$.
Clearly we may assume that $G$ is chosen vertex-minimal with respect to these properties.
That is, condition~(\ref{item:forbidden3}) holds for $G$.
But then the order of $G$ is at most $N$, and thus $G \in \mathcal F'$, a contradiction.

Now we assume that $G$ has a square root in $\mathcal G$.
By Lemma~\ref{lem:simple}, $|\bigcap \mathcal C(G)| \ge |\mathcal C(G)|$ and $T(G) \in \mathcal G$.
In particular, it remains to prove that $G$ is $\mathcal F'$-free.
Suppose that $G$ contains some graph $F \in \mathcal F'$ as an induced subgraph.
Then $T(F)$ is (isomorphic to) an induced subgraph of $T(G)$, due to the definition of the trunk.
Since $T(G) \in \mathcal G$ and $\mathcal G$ is a hereditary class, $T(F) \in \mathcal G$.
However,~(\ref{item:forbidden2}) implies that $F$ does not have a square root in $\mathcal G$.
Thus, Lemma~\ref{lem:simple} and (\ref{item:forbidden1}) imply $T(F) \notin \mathcal G$, a contradiction.
%Let $\mathcal G$ be a simple class of 3-sun-free split graphs, defined by a finite set $\mathcal F$ of forbidden induced subgraphs.
%There is a finite set $\mathcal F'$ of graphs such that the following holds:
%a given connected graph $G$ has a square root in $\mathcal G$ if and only if $|\bigcap \mathcal C(G)| \ge |\mathcal C(G)|$ and $G$ is $\mathcal F'$-free.
\end{proof}

\begin{proof}[Proof of Theorem~\ref{thm:odd-sun-free-characterization}]
Let $G$ be a graph that has a connected odd-sun-free split square root $H$.
Let $C \cup I$ be a split partition of $H$.

Since $H$ is 3-sun-free, we know that $\big|\bigcap \mathcal{C}(G)\big|\ge |\mathcal{C}(G)|$, by Proposition~\ref{prop:2}.
It remains to show that $G$ is balanced.
For this, let $B$ be the bipartite graph on the vertex set $C \cup I$ containing exactly those edges of $H$ that join a vertex of $C$ to a vertex of $I$.
As $H$ is odd-sun-free, $B$ does not contain induced cycles of length $2k$, for any odd $k$.

By Proposition~\ref{prop:2}, for every $Q \in \C(G)$ there is some $x_Q \in C$ with $N_H(x_Q) \cap I = Q \cap I$.
Thus consider the hypergraph $\mathcal H$ on the vertex set $I$ with hyperedges $\{N_H(x_Q) \cap I : Q \in \C(G)\}$.
Note that the vertex-to-hyperedge incidence graph of $\mathcal H$ is an induced subgraph of $B$.
Hence, $\mathcal H$ is balanced.
But $\mathcal H$ is exactly the maximal clique hypergraph of $G[I]$, since all vertices in $V(G) \setminus I = C$ are universal vertices in $G$.
Consequently, $G[I]$ is balanced.
As $G$ is obtained from $G[I]$ by simply attaching the vertices of $C$ as universal vertices, $G$ is balanced, too.

Now let a balanced graph $G$ with $\big|\bigcap \mathcal{C}(G)\big|\ge |\mathcal{C}(G)|$ be given.
Let $B$ be the vertex-to-maximal-clique incidence graph of $G$.
As $G$ is balanced, $B$ does not contain induced cycles of length $2k$, for any odd $k$.
Thus, $T(G)$ is odd-sun-free.
Recall from the proof of Corollary~\ref{cor:classes} that the class of odd-sun-free split graphs is simple.
By Lemma~\ref{lem:simple}, $G$ has an odd-sun-free square split root $H$ as desired.

Since balanced graphs can be recognized in polynomial time~\cite{BDLS06} and have a linear number of maximal cliques, we may thus decide wether a given graph admits an odd-sun-free split square root in polynomial time.
\end{proof}

\section{Conclusion}\label{sec:conclusion}
%%%%%%%%%%%%%%%%%%%%%%%%%%%%%%%%%%%%%%%%%%%

In this paper we discussed the complexity of recognizing squares of hereditary subclasses $\mathcal G$ of 3-sun-free split graphs.
Our Theorem~\ref{thm:recognition} shows that if such a class $\mathcal G$ is simple, squares of graphs in $\mathcal G$ can be recognized in polynomial time if $\mathcal G$ can be recognized in polynomial time.
In Corollary~\ref{cor:classes} we gave several examples for such simple classes, including, according to our knowledge, all previously known subclasses of split graphs for which the square graph problem is known to be solvable in polynomial time.

Theorem~\ref{thm:simple-and-finite} shows that, for any simple graph class $\mathcal G$ of 3-sun-free split graphs defined by finitely many forbidden induced subgraphs, there exists a finite list of forbidden induced subgraphs for the class of squares of graphs of $\mathcal G$.
%An example of such a characterization is given in Theorem~\ref{thm:nets3free}, where we characterized squares of $(S_3,\mbox{net})$-free split graphs, the largest self-complementary hereditary subclass of 3-sun-free split graphs.
In addition, we gave a characterization of squares of odd-sun-free split graphs, in Theorem~\ref{thm:odd-sun-free-characterization}.
The proof shows how our tools work even in the case of graph classes defined by infinitely many forbidden induced subgraphs.

In view of Theorem~\ref{thm:simple-and-finite}, we have the following open problem.

\begin{question*}\label{que:finite}
Let $\mathcal G$ be a class of $3$-sun-free split graphs, defined by a finite set $\mathcal F$ of forbidden induced subgraphs.
Is there a finite set $\mathcal F'$ of graphs such that the following holds:
A given connected graph $G$ has a square root in $\mathcal G$ if and only if $|\bigcap \mathcal C(G)| \ge |\mathcal C(G)|$ and $G$ is $\mathcal F'$-free?
\end{question*}
Indeed, Theorem~\ref{thm:simple-and-finite} shows that this question is true in the case that $\mathcal G$ is simple.
If it is true in general, we could derive the following, which we consider as a very interesting result:
for every class $\mathcal G$ of 3-sun-free split graphs defined by finitely many forbidden induced subgraphs there is a polytime algorithm to decide whether a given graph has a square root in $\mathcal G$.
This is due to the well-known fact that graph classes defined by finitely many forbidden induced subgraphs can be recognized in polynomial time.

Another natural question is whether one can give a general framework to find characterizations such as Theorem~\ref{thm:odd-sun-free-characterization} for all simple classes.
Such a framework would perfectly complement Theorem~\ref{thm:recognition} which is, in a sense, such a general framework but from an algorithmic perspective.
In a first step, one could try to extend Theorem~\ref{thm:simple-and-finite} from the existence statement to a construction manual for obtaining the promised list of forbidden induced subgraphs.

\bibliographystyle{amsplain}
\bibliography{graph-powers}

%%%%%%%%%%%%%%%%%%%%%%%%%%%%%%%%%%%%%%%%%%%

%\clearpage
\newpage

\begin{appendix}

%%%%%%%%%%%%%%%%%%%%%%%%%%%%%%%%%%%%%%%%%%%

\section*{Appendix: Graph Classes} For the sake of completeness, definitions and relevant information on graph classes listed in Corollary~\ref{cor:classes} will be given here. We derive that these classes are simple with respect to $f$ defined in the proof of Corollary~\ref{cor:classes}: for $H=(C\cup I, E)\in \cal G$, $f(H)$ is obtained from $H$ by adding a new vertex and making it be adjacent to all vertices in the clique $C$ of $H$.

It is obvious that $f(H)^2=H^2\oplus K_1$. It remains to show that $f(H)\in\cal G$ for all classes $\cal G$ listed in Corollary~\ref{cor:classes}, and this should be clear in case $\cal G$ is the class of $3$-sun-free split graphs, ($3$-sun, net)-free split graphs, strongly chordal split graphs, odd-sun-free split graphs, respectively.

Recall that the graph \emph{net} is the complement of the $3$-sun $S_3$. We write $S_\ell^-$ for the graph obtained from the $\ell$-sun $S_\ell$ by deleting a degree-$2$ vertex, and co-$S_\ell^-$ for the complement of $S_\ell^-$.

\medskip\noindent
\textbf{Interval split graphs.}\,
A graph is an \emph{interval split graph} if it is an interval graph and a split graph at the same time. Since interval graphs as well as split graphs can be recognized in linear time (see, for instance, \cite{BLS99,isgci,Spinrad}), recognizing interval split graphs can be done in linear time. 
Interval split graphs have been characterized by Foldes and Hammer in~\cite{FolHam77} as follows: 

\begin{quote}
$G$ is an interval split graph if and only if $G$ is a $(S_3, \mbox{net}, S_4^-)$-free split graph.
\end{quote}

Form this characterization, it is obvious that $f(H)$ is an interval split graph whenever $H$ is an interval split graph.

\medskip\noindent
\textbf{Comparability split graphs.}\,
A graph is a \emph{comparability split graph} if it is a comparability graph and a split graph at the same time. Since comparability graphs can be recognized in time proportional to matrix multiplication (see, for instance, \cite{BLS99,isgci,Spinrad}), recognizing comparability split graphs can be done in polynomial time. 
Comparability split graphs have been characterized by F\"oldes and Hammer in~\cite{FolHam77a} as follows: 

\begin{quote}
$G$ is a comparability split graph if and only if $G$ is a $(S_3, \mbox{net}, \mbox{co-}S_4^-)$-free split graph.
\end{quote}

Form this characterization, it is obvious that $f(H)$ is a comparability split graph whenever $H$ is a comparability split graph.

\medskip\noindent
\textbf{Permutation split graphs.}\, %(in $\mathcal O(\max\{n^2m,m^2\})$ time);
A graph is a \emph{permutation split graph} if it is a permutation graph and a split graph at the same time. Since permutation graphs can be recognized in linear time (see, for instance, \cite{BLS99,isgci,Spinrad}), recognizing permutation split graphs can be done in linear time. 
Since permutation graphs are exactly those graphs that are both comparability graphs and co-comparability graphs, we have: 

\begin{quote}
$G$ is a permutation split graph if and only if $G$ is a $(S_3, \mbox{net}, S_4^-, \mbox{co-}S_4^-)$-free split graph.
\end{quote}

Form this characterization, it is obvious that $f(H)$ is a permutation split graph whenever $H$ is a permutation split graph.

\medskip\noindent
\textbf{Probe threshold split graphs.}\, 
Let $\mathcal C$ be a graph class. A graph $G$ is called a probe $\mathcal{C}$ graph if there is an independent set $N$ of \lq non-probe\rq\ vertices in $G$ such that we can add some new edges between certain non-probe vertices and obtain a new graph $G'$ in the class $\mathcal{C}$. A graph is a \emph{probe threshold split graph} if it is a probe threshold graph and a split graph at the same time. Probe threshold graphs have been characterized in~\cite{BLR09} and can be recognized in linear time; see also~\cite{isgci}. Therefore, probe threshold split graphs can be recognized in linear time. Moreover, it turns out that probe threshold split graphs can be characterized as $(S_3,\mbox{net})$-free split graphs without the four additional forbidden induced subgraphs seen in Fig.~\ref{fig:G4789}, as was shown in~\cite{BLR09}.

From this characterization, it is obvious that that $f(H)$ is a probe threshold split graph whenever $H$ is a probe threshold split graph. 

\begin{figure}[ht] %b]
\begin{center}
\psset{unit=0.6}
\begin{pspicture}(0,0)(17,4)

\cnode(0,2){0.1cm}{a_1}
\cnode(1,3){0.1cm}{a_2}
\cnode(1,1){0.1cm}{a_3}
\cnode(2.5,3){0.1cm}{a_4}
\cnode(2.5,1){0.1cm}{a_5}
\cnode(3.5,2){0.1cm}{a_6}

\ncarc[arcangle=0]{-}{a_1}{a_2}
\ncarc[arcangle=0]{-}{a_1}{a_3}
\ncarc[arcangle=0]{-}{a_2}{a_3}

\ncarc[arcangle=0]{-}{a_2}{a_4}
\ncarc[arcangle=0]{-}{a_2}{a_5}
\ncarc[arcangle=0]{-}{a_3}{a_4}
\ncarc[arcangle=0]{-}{a_3}{a_5}
\ncarc[arcangle=0]{-}{a_4}{a_5}

\ncarc[arcangle=0]{-}{a_4}{a_6}
\ncarc[arcangle=0]{-}{a_5}{a_6}

\cnode(6.5,3){0.1cm}{b_1}
\cnode(5,2){0.1cm}{b_2}
\cnode(5.75,1){0.1cm}{b_3}
\cnode(7.25,1){0.1cm}{b_4}
\cnode(8,2){0.1cm}{b_5}
\cnode(5.25,0){0.1cm}{b_6}
\cnode(7.75,0){0.1cm}{b_7}

\ncarc[arcangle=0]{-}{b_1}{b_2}
\ncarc[arcangle=0]{-}{b_1}{b_3}
\ncarc[arcangle=0]{-}{b_1}{b_4}
\ncarc[arcangle=0]{-}{b_1}{b_5}

\ncarc[arcangle=0]{-}{b_2}{b_3}
\ncarc[arcangle=0]{-}{b_3}{b_4}
\ncarc[arcangle=0]{-}{b_4}{b_5}

\ncarc[arcangle=0]{-}{b_3}{b_6}
\ncarc[arcangle=0]{-}{b_4}{b_7}

\cnode(11,3){0.1cm}{c_1}
\cnode(9.5,2){0.1cm}{c_2}
\cnode(10.5,2){0.1cm}{c_3}
\cnode(11.5,2){0.1cm}{c_4}
\cnode(12.5,2){0.1cm}{c_5}
\cnode(11,1){0.1cm}{c_6}
\cnode(11,0){0.1cm}{c_7}

\ncarc[arcangle=0]{-}{c_1}{c_2}
\ncarc[arcangle=0]{-}{c_1}{c_3}
\ncarc[arcangle=0]{-}{c_1}{c_4}
\ncarc[arcangle=0]{-}{c_1}{c_5}
\ncarc[arcangle=0]{-}{c_1}{c_6}

\ncarc[arcangle=0]{-}{c_2}{c_3}
\ncarc[arcangle=0]{-}{c_3}{c_4}
\ncarc[arcangle=0]{-}{c_4}{c_5}

\ncarc[arcangle=0]{-}{c_2}{c_6}
\ncarc[arcangle=0]{-}{c_3}{c_6}
\ncarc[arcangle=0]{-}{c_4}{c_6}
\ncarc[arcangle=0]{-}{c_6}{c_7}

\cnode(15.5,3){0.1cm}{d_1}
\cnode(14,2){0.1cm}{d_2}
\cnode(15,2){0.1cm}{d_3}
\cnode(16,2){0.1cm}{d_4}
\cnode(17,2){0.1cm}{d_5}
\cnode(15.5,1){0.1cm}{d_6}
\cnode(15.5,0){0.1cm}{d_7}
\cnode(15.5,4){0.1cm}{d_8}

\ncarc[arcangle=0]{-}{d_1}{d_2}
\ncarc[arcangle=0]{-}{d_1}{d_3}
\ncarc[arcangle=0]{-}{d_1}{d_4}
\ncarc[arcangle=0]{-}{d_1}{d_5}
\ncarc[arcangle=0]{-}{d_1}{d_6}
\ncarc[arcangle=0]{-}{d_1}{d_8}

\ncarc[arcangle=0]{-}{d_2}{d_3}
\ncarc[arcangle=0]{-}{d_3}{d_4}
\ncarc[arcangle=0]{-}{d_4}{d_5}

\ncarc[arcangle=0]{-}{d_2}{d_6}
\ncarc[arcangle=0]{-}{d_3}{d_6}
\ncarc[arcangle=0]{-}{d_4}{d_6}
\ncarc[arcangle=0]{-}{d_6}{d_7}

\ncarc[arcangle=0]{-}{d_5}{d_6}

\end{pspicture}
\end{center}
\caption{The four additional forbidden subgraphs for probe threshold split graphs.}
\label{fig:G4789}
\end{figure}
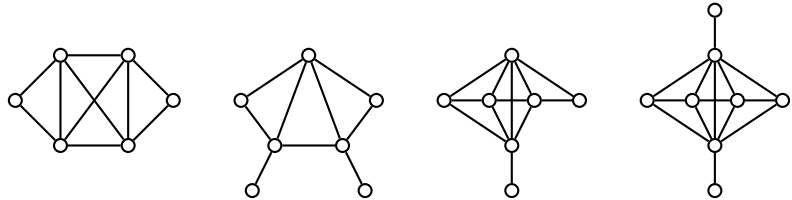

\end{appendix}

\end{document}